\documentclass{elsarticle}
\usepackage{amsmath,amssymb,latexsym,amsmath,enumerate,verbatim,amsfonts,mathrsfs,bm}

\newtheorem{lemma}{Lemma}[section]
\newtheorem{example}{Example}[section]

\numberwithin{equation}{section}
\begin{document}
	
\title{Phase Angle and Effective Second-Harmonic Generation Coefficient \thanks{This  work was supported by
		the National Natural Science Foundation of P.R. China (Grant No.12171064).}}

%\author[1]{Die Xiao}
\author[1]{Yisheng Song\corref{mycorrespondingauthor}}\cortext[mycorrespondingauthor]{Corresponding author.  E-mail address: yisheng.song@cqnu.edu.cn (Yisheng Song).}

\affiliation[1]{organization={School of Mathematical Sciences, Chongqing Normal University},
	city={Chongqing},
	postcode={401331},
%	%            state={P.R.}
	country={P.R. China}\\ { Email:  yisheng.song@cqnu.edu.cn}}

\begin{abstract} In this paper, the calculation formulae of phase angle are given for two classes of largest effective SHG coefficients in uniaxial crystals by means of  the optimization theory.  With the help of such calculation formulae, we present the best phase angles and azimuth angles  of all  uniaxial crystals class,   as well as their effective SHG coefficients. Furthermore, these calculation  is only dependent upon some principal subtensor of second order susceptibility tensor. 
\end{abstract}

\maketitle
\section{Introduction}
The Second-Harmonic Generation (SHG) of light was observed in 1961 by Franken, et al.\cite{FHPW1961}, which is  the first non-linear optical effect \cite{BFHP1962,BFWW1962,FHPW1961,FW1963}. In nonlinear optics, the electric polarization $\vec{\bm P}$ may be said as, $$\vec{\bm P}=\varepsilon_0{\bm\chi}^{(1)}\vec{\bm E}+\varepsilon_0{\bm\chi}^{(2)}\vec{\bm E}\vec{\bm E}+\varepsilon_0{\bm\chi}^{(3)}\vec{\bm E}\vec{\bm E}\vec{\bm E}+\cdots,$$ 
The above equation in component form reads as follows:
$${\bm P}_i=\varepsilon_0\sum_{j=1}^3t_{ij}{\bm E}_j+\varepsilon_0\sum_{j,k=1}^3t_{ijk}{\bm E}_j{\bm E}_k+\varepsilon_0\sum_{j,k,l=1}^3t_{ijkl}{\bm E}_j{\bm E}_k{\bm E}_l+\cdots,$$ 
where $\varepsilon_0$ is the vacuum  permittivity,  ${\bm P}_i$ is the $i-$th Cartesian coordinate of the polarization $\vec{\bm P}$,   $\vec{\bm E}$ is the optical electric field  within the medium, ${\bm\chi}^{(1)}=(t_{ij})$ is the linear dielectric susceptibility, ${\bm\chi}^{(2)}=(t_{ijk})$ is the second-order nonlinear susceptibility, ${\bm\chi}^{(3)}=(t_{ijkl})$ is the third-order nonlinear susceptibility, and so on.   The second-order nonlinear susceptibility tensor ${\bm\chi}^{(2)}=(t_{ijk})$ may be  generated the second-order nonlinear polarization $\vec{\bm P}^{(2)}$ \cite{B2020,CLH1992,J1970,S2024}, \begin{equation}\label{eq:11}\
	\vec{\bm P}^{(2)}=\varepsilon_0{\bm\chi}^{(2)}\vec{\bm E}\vec{\bm E}=\left(\varepsilon_0\sum_{j,k=1}^3t_{ijk}{\bm E}_j{\bm E}_k\right).
\end{equation} The tensor ${\bm\chi}^{(2)}=(t_{ijk})$ is  a  piezoelectric-type  tensor. That is, it is the  intrinsic permutation symmetry in the  last two indices, $$t_{ijk}=t_{ikj}.$$ 

 The second-order susceptibility tensor ${\bm\chi}^{(2)}$ definited six types of  frequency-doubled lights, $oo\to e$, $oe\to o$, $eo\to o$, $eo\to e$, $oe\to e$, $ee\to o$.  Then four different effective SHG coefficients are involved, ${\bm \chi}_{\text{eff}}^{oo-e}$, ${\bm \chi}_{\text{eff}}^{oe-o}={\bm \chi}_{\text{eff}}^{eo-o}$, ${\bm \chi}_{\text{eff}}^{eo-e}={\bm \chi}_{\text{eff}}^{oe-e}$, ${\bm \chi}_{\text{eff}}^{ee-o}.$  Their exact value depends on the components of second-order susceptibility tensor ${\bm\chi}^{(2)}$, and may play an important  role in  classical crystals and quantum optical applications, such as evaluating performance, optimizing devices, designing new nonlinear optical crystals \cite{B2020,C2012,CJ2000,CLH1992,DGGAM2009,GAS2006,PRK2025,WC1992,ZYZ2024}.  

  For solving largest effective SHG coefficients,  some tensors theory and methods, especially, piezoelectric-type tensors theory,  may be involved. For piezoelectric-type tensor,  Chen-J\'akli-Qi \cite{CJQ2023} first used the concepts of  C-eigenvalues, and gave  its applictions in  piezoelectric materials  and liquid crystal physics. Recently, Xiao-Song \cite{XS2026} presented the relationship between the largest effective SHG coefficients and the largest C-eigenvalue of susceptibility tensor, and  cutted down variables of the largest effective SHG coefficients in uniaxial crystals.  Many numerical techniques were applied to compute the largest  C-eigenvalue of a piezoelectric-type tensor \cite{CCW2019, LY2021,LLJ2025,LM2022,WCW2020,WCW2023,YL2022,ZLS2023,ZL2022}. 
  
In this paper,  we mainly provide the analytical formulae of phase angle definited by the largest effective SHG coefficients in uniaxial crystals.  Furthermore, applying these calculation formulae, the best phase angles and azimuth angles  are given for all  uniaxial crystals classes.  At the same time,  the largest  effective SHG coefficients are exactually established. 

\section{Two Crucial Lemmas}

 We know the largest eigenvalue of a squre matrix is equal to the maximum of a quadratic form given by such a matrix  on the unit sphere \cite{HJ202}.   Then the following lemma is easy to be showed (see Appendix  A for detailed proof ).

\begin{lemma}\label{lem:21} Let $f({\bm x})=a_{11}x_1^2+2a_{12}x_1x_2+a_{22}x_2^2$ be a homogeneous polynomial. Then
	$$\max\{f({\bm x}): x_1^2+x_2^2=1\} = \dfrac12\left(a_{11}+a_{22}+\sqrt{(a_{11}-a_{22})^2+4a_{12}^2}\right)$$
	with  its maximizer  ${\bm x}^*=(x^*_1,x^*_2)^\top,$ $${\bm x}^*=\begin{cases}
	\pm\left(\sqrt{\frac{1+\alpha}2}, \sqrt{\frac{1-\alpha}2}\right)^\top, & a_{12}\ge0,\\
	\pm\left(-\sqrt{\frac{1+\alpha}2},\sqrt{\frac{1-\alpha}2}\right)^\top, &a_{12}\le0,\\
	(y_1,y_2)^\top \mbox{ with }y_1^2+y_2^2=1, &a_{12}=0\mbox{ and }a_{11}=a_{22}, 
	\end{cases}$$  where $\alpha=\frac{a_{11}-a_{22}}{\sqrt{(a_{11}-a_{22})^2+4a_{12}^2}}.$ \end{lemma}

In  a Hibert space with inner product $\langle\cdot, \cdot\rangle$ and norm $\|\cdot\|$,  it is well-known that Cauchy-Schwarz inequality implies $$|\langle x,y\rangle|\leq \|x\|\|y\|,$$
and moreover,  such a  equality holds if and only if $y=\alpha x$ for some $\alpha$ \cite{R1991}. So the following lemma is easy to obtain.

\begin{lemma}\label{lem:22}
	Let $\mathbb{H}$ be a Hibert space with inner product $\langle\cdot, \cdot\rangle$ and the norm $\|\cdot\|$ induced by an inner product. Then  for fixed ${\bm x}\in \mathbb{H},$
	$$\max_{\|{\bm y}\|=1}\langle {\bm x},{\bm y}\rangle=\|{\bm x}\| \mbox{ with a maximizer } {\bm y}^*=\frac{{\bm x}}{\|{\bm x}\|}.$$ 
	$$\max_{\|{\bm y}\|=1}|\langle {\bm x},{\bm y}\rangle|=\|{\bm x}\| \mbox{ with a maximizer } {\bm y}^*=\pm\frac{{\bm x}}{\|{\bm x}\|}.$$ In particular, 
	$$\max_{\|{\bm y}\|_2=1}(a_1y_1+a_2y_2+\cdots +a_ny_n)=\sqrt{|a_1|^2+|a_2|^2+\cdots+|a_n|^2}$$  with a maximizer $$ {\bm y}^*=\frac{(a_1,a_2,\cdots,a_n)^H}{\sqrt{|a_1|^2+|a_2|^2+\cdots+|a_n|^2}},$$
	here $H$ is conjugate transpose and $\|{\bm x}\|_2=\sqrt{|x_1|^2+|x_2|^2+\cdots+|x_n|^2}$.
\end{lemma}

\section{Phase angle and  Azimuth angle}
Let ${\bm x}$ and ${\bm y}$ be  unit vectors of polarization direction of $o$-light (ordinary light)  and $e$-light (extraordinary light) in uniaxial crystals,  respectively. That is, \begin{equation}\label{eq:31}
	{\bm x}=(x_1,x_2,x_3)^\top=(\sin\varphi,-\cos\varphi,0)^\top,
\end{equation}
\begin{equation}\label{eq:32}
	{\bm y}=(y_1,y_2,y_3)^\top=(-\cos\theta\cos\varphi,-\cos\theta\sin\varphi,\sin\theta)^\top,
\end{equation}
where $\theta$  and $\varphi$ are, respectively, the phase angle and  the azimuth angle of wave vector.
Let ${\bm z}=(z_1,z_2)=(\cos\theta, \sin\theta)$. Then ${\bm y}$  may be rewritten as follows,
\begin{equation}\label{eq:33}
	{\bm y}=(x_2z_1,-x_1z_1,z_2)^\top.
\end{equation}

Assume that ${\bm\chi}^{(2)}=(t_{ijk})\in\mathbb{R}^{3\times3\times3}$ is a nonlinear second order susceptibility tensor.  Then two effective SHG coefficients may be  written as 
\[
{\bm \chi}_{\text{eff}}^{oo-e}{\bm y}{\bm x}{\bm x}={\bm\chi}^{(2)}{\bm y}{\bm x}{\bm x}=\sum_{i,j,k=1}^3t_{ijk}y_ix_jx_k=\sum_{i=1}^3\sum_{j,k=1}^2t_{ijk}y_ix_jx_k,
\]
\[
{\bm \chi}_{\text{eff}}^{ee-o}{\bm x}{\bm y}{\bm y}={\bm\chi}^{(2)}{\bm x}{\bm y}{\bm y}=\sum_{i,j,k=1}^3t_{ijk}x_iy_jy_k=\sum_{i=1}^2\sum_{j,k=1}^3t_{ijk}x_iy_jy_k.
\]
So, the effective SHG coefficient ${\bm \chi}_{\text{eff}}^{oo-e}{\bm y}{\bm x}{\bm x}$ is independent of 15 components of   second order susceptibility tensor ${\bm\chi}^{(2)}=(t_{ijk})$, $$t_{113}=t_{131}, t_{123}=t_{132}, t_{133}, t_{213}=t_{231}, t_{223}= t_{232}, t_{233}, t_{313}=t_{331}, t_{323}=t_{332}, t_{333}.$$ 
${\bm \chi}_{\text{eff}}^{oo-e}{\bm y}{\bm x}{\bm x}$ is only dependent upon a principal subtensor of  ${\bm\chi}^{(2)}$, $${\bm\chi}^{(2)}_{\textbf{sub-1}}=(t_{ijk})\in\mathbb{R}^{3\times2\times2}.$$
The concept of principal subtensor sees Ref. \cite{SQ2015}.
The effective SHG coefficient ${\bm \chi}_{\text{eff}}^{ee-o}{\bm x}{\bm y}{\bm y}$ is independent of 9 components of   second order susceptibility tensor ${\bm\chi}^{(2)}=(t_{ijk})$, $$t_{311}, t_{312}=t_{321},  t_{313}= t_{331}, t_{323}=t_{332}, t_{322}, t_{333}.$$ ${\bm \chi}_{\text{eff}}^{ee-o}{\bm x}{\bm y}{\bm y}$ is only dependent upon a principal subtensor of  ${\bm\chi}^{(2)}$, $${\bm\chi}^{(2)}_{\textbf{sub-2}}=(t_{ijk})\in\mathbb{R}^{2\times3\times3}.$$
Clearly, ${\bm \chi}_{\text{eff}}^{eo-o}{\bm x}{\bm y}{\bm x}={\bm \chi}_{\text{eff}}^{oe-o}{\bm x}{\bm x}{\bm y}$ and ${\bm \chi}_{\text{eff}}^{eo-e}{\bm y}{\bm y}{\bm x}={\bm \chi}_{\text{eff}}^{oe-e}{\bm y}{\bm x}{\bm y}$  respectively depends upon   principal subtensors of  ${\bm\chi}^{(2)}$, $${\bm\chi}^{(2)}_{\textbf{sub-3}}=(t_{ijk})\in\mathbb{R}^{2\times3\times2}\mbox{ and }{\bm\chi}^{(2)}_{\textbf{sub-4}}=(t_{ijk})\in\mathbb{R}^{3\times2\times3}.$$

It follows from Theorem 3.1 of Xiao-Song \cite{XS2026} that \begin{equation}\label{eq:34}{\bm \chi}_{\text{eff}}^{oo-e}{\bm y}{\bm x}{\bm x}=z_1\left(x_2V_1({\bm x})-x_1V_2({\bm x})\right)+z_2V_3({\bm x}),\end{equation}where  $V _i({\bm x})=\sum\limits_{j,k=1}^2 t_{ijk}x_jx_k$ for $i=1,2,3$.
Apply Lemma \ref{lem:22} to Eq. \eqref{eq:34}, we have 
\begin{equation}\label{eq:35}\aligned
	{\bm \chi}^{\text{eff}}_1=&\max{\bm \chi}_{\text{eff}}^{oo-e}{\bm y}{\bm x}{\bm x}\\=&\max\left\{\sqrt{\left(x_2V_1({\bm x})-x_1V_2({\bm x})\right)^2+(V_3({\bm x}))^2}: x_1^2+x_2^2=1\right\}\endaligned
\end{equation}
with its maximizers ${\bm x}^*=(\sin\varphi^*,-\cos\varphi^*,0)$ and \begin{equation}\label{eq:36}{\bm z}^*=(\cos\theta^*,\sin\theta^*)=\frac{(x^*_2V_1({\bm x}^*)-x^*_1V_2({\bm x}^*), V_3({\bm x}^*))}{\sqrt{\left(x^*_2V_1({\bm x}^*)-x^*_1V_2({\bm x}^*)\right)^2+(V_3({\bm x}^*))^2}}.\end{equation}

Solve the equation, ${\bm \chi}_{\text{eff}}^{oo-e}{\bm y}{\bm x}{\bm x}=0$ to yield 
\begin{equation}\label{eq:37} \tan\theta=\begin{cases}
	\dfrac{x_1V_2({\bm x})-x_2V_1({\bm x})}{V_3({\bm x})},V_3({\bm x})\ne0,\\
	\pm\infty, V_3({\bm x})=0,
	\end{cases}\end{equation} and so, Eq. \eqref{eq:37} means ${\bm \chi}_{\text{eff}}^{oo-e}{\bm y}{\bm x}{\bm x}=0.$\\

By Theorem 3.2 of Xiao-Song \cite{XS2026}, we obtain
	\begin{equation}\label{eq:38}
		{\bm \chi}_{\text{eff}}^{ee-o}{\bm x}{\bm y}{\bm y}=z^2_1a_{11}({\bm x})+2z_1z_2a_{12}({\bm x})+z^2_2a_{22}({\bm x}),
	\end{equation}
	where	$a_{11}({\bm x})=(t_{111}-2t_{212})x_1x^2_2+(t_{222}-2t_{112})x^2_1x_2+t_{122}x^3_1+t_{211}x^3_2,$
	$$a_{12}({\bm x})=(t_{113}-t_{223})x_1x_2-t_{123}x^2_1+t_{213}x^2_2, a_{22}({\bm x})=t_{133}x_1+t_{233}x_2.$$ 

Applying Lemma \ref{lem:21} to Eq. \eqref{eq:38}, we have 
	\begin{equation}\label{eq:39}
		\aligned &{\bm \chi}^{\text{eff}}_2=\max{\bm \chi}_{\text{eff}}^{ee-o}{\bm x}{\bm y}{\bm y}\\
		=&\dfrac12\max_{\|{\bm x}\|_2=1}\left\{a_{11}({\bm x})+a_{22}({\bm x})+\sqrt{\left(a_{11}({\bm x})-a_{22}({\bm x})\right)^2+4(a_{12}({\bm x}))^2}\right\}\endaligned
	\end{equation} with its maximizers ${\bm x}^*=(\sin\varphi^*,-\cos\varphi^*,0)$ and \begin{equation}\label{eq:310}{\bm z}^*=(\cos\theta^*,\sin\theta^*)=\begin{cases}
	\pm\left(\sqrt{\frac{1+\alpha}2}, \sqrt{\frac{1-\alpha}2}\right),  a_{12}({\bm x}^*)\ge0,\\
	\pm\left(-\sqrt{\frac{1+\alpha}2},\sqrt{\frac{1-\alpha}2}\right), a_{12}({\bm x}^*)\le0,\\
	(\cos\theta,\sin\theta), a_{12}({\bm x}^*)=0, a_{11}({\bm x}^*)=a_{22}({\bm x}^*) 
	\end{cases}\end{equation}  where $\alpha=\frac{a_{11}({\bm x}^*)-a_{22}({\bm x}^*)}{\sqrt{(a_{11}({\bm x}^*)-a_{22}({\bm x}^*))^2+4(a_{12}({\bm x}^*))^2}}. $
Meanwhile, from Eq. \eqref{eq:38}, it follows that $$\frac{{\bm \chi}_{\text{eff}}^{ee-o}{\bm x}{\bm y}{\bm y}}{z_1^2}=a_{22}({\bm x})\left(\frac{z_2}{z_1}\right)^2 +2 a_{12}({\bm x})\cdot \frac{z_2}{z_1}+a_{11}({\bm x}),$$  and hence, solving an equation, ${\bm \chi}_{\text{eff}}^{ee-o}{\bm x}{\bm y}{\bm y}=0$ to yield
\begin{equation}\label{eq:311}
	 \tan\theta=\begin{cases}
	\frac{-a_{12}({\bm x})\pm\sqrt{(a_{12}({\bm x}))^2-a_{11}({\bm x})a_{22}({\bm x})}}{a_{22}({\bm x})},  a_{22}({\bm x})\ne0,\\
	-\frac{a_{11}({\bm x})}{2a_{12}({\bm x})},  \ \  \ \ \ \ \ \ \ \ \ \ \ \ \ \ \ \ \  a_{22}({\bm x})=0,  a_{12}({\bm x})\ne0,\\
	+\infty \mbox{ or }-\infty,  \ \ \ \ \  \ \ \ \ \ \ \ \ a_{22}({\bm x})=a_{12}({\bm x})=0, a_{11}({\bm x})\ne0,\\
	0\mbox{ or }+\infty\mbox{ or }-\infty,  \ \ \ \ \ \ a_{11}({\bm x})=a_{22}({\bm x})=0,  a_{12}({\bm x})\ne0.
	\end{cases} 
\end{equation} So, Eq. \eqref{eq:311} implies ${\bm \chi}_{\text{eff}}^{ee-o}{\bm x}{\bm y}{\bm y}=0.$\\

With thw help of Eqs. \eqref{eq:34}--\eqref{eq:311}, we may easily calculate  the phase angle and the azimuth angle and effective SHG coefficient. For detail calculated process, see Appendix B.\\

\textbf{1. Point Group $\bar{4}$ (e.g., Co$_4$B$_2$As$_2$O$_{12}$·4H$_2$O, HgGa$_2$S$_4$)}  $$ \theta=0, \pi \Longrightarrow {\bm \chi}_{\text{eff}}^{oo-e}{\bm y}{\bm x}{\bm x}=0$$  or $$\varphi=-\frac12\arctan\frac{d_{31}}{d_{36}}, \pi-\frac12\arctan\frac{d_{31}}{d_{36}} \Longrightarrow {\bm \chi}_{\text{eff}}^{oo-e}{\bm y}{\bm x}{\bm x}=0. . $$ 

$ {\bm \chi}^{\text{eff}}_1=\max{\bm \chi}_{\text{eff}}^{oo-e}{\bm y}{\bm x}{\bm x}=\sqrt{d^2_{36}+d_{31}^2}$ with its corresponding azimuth angle $\varphi^*$ and  phase angle $\theta^*$, $$\varphi^*=\frac12\arcsin \frac{d_{36}}{\sqrt{d^2_{36}+d_{31}^2}}, \theta^*=\dfrac{3\pi}2 \mbox{ or }\varphi^*=\dfrac{\pi}2+\frac12\arcsin \frac{d_{36}}{\sqrt{d^2_{36}+d_{31}^2}}, \theta^*=\dfrac{\pi}2 .$$

 $$ \theta=0, \pi,\frac\pi2, \frac{3\pi}2 \Longrightarrow {\bm \chi}_{\text{eff}}^{ee-o}{\bm x}{\bm y}{\bm y}=0,$$ or $$ \varphi=\frac12\arctan\frac{d_{14}}{d_{15}}, \pi+\frac12\arctan\frac{d_{14}}{d_{15}} \Longrightarrow {\bm \chi}_{\text{eff}}^{ee-o}{\bm x}{\bm y}{\bm y}=0.$$
 
 ${\bm \chi}^{\text{eff}}_2=\max{\bm \chi}_{\text{eff}}^{ee-o}{\bm x}{\bm y}{\bm y}=\sqrt{d^2_{14}+d_{15}^2}$ with its corresponding azimuth angle $\varphi^*$ and  phase angle $\theta^*$, $$\varphi^*=\frac\pi2-\frac12\arcsin \frac{d_{15}}{\sqrt{d^2_{14}+d_{15}^2}}, \theta^*= \dfrac{3\pi}4\mbox{ or }\dfrac{7\pi}4;$$
 $$\varphi^*=\pi-\frac12\arcsin \frac{d_{15}}{\sqrt{d^2_{14}+d_{15}^2}}, \theta^*= \dfrac{\pi}4\mbox{ or }\dfrac{5\pi}4.$$

\textbf{2. Point Group $\bar{4}2$m (e.g., KH$_2$PO$_4$, NH$_4$H$_2$PO$_4$, CH$_4$N$_2$O)}
 $$ \theta=0, \pi \Longrightarrow {\bm \chi}_{\text{eff}}^{oo-e}{\bm y}{\bm x}{\bm x}=0,\mbox{ or } \varphi=0, \frac\pi2, \frac{3\pi}2, \pi  \Longrightarrow {\bm \chi}_{\text{eff}}^{oo-e}{\bm y}{\bm x}{\bm x}=0.$$ 
 
 ${\bm \chi}_1^{\textbf{eff}}=\max{\bm \chi}_{\text{eff}}^{oo-e}{\bm y}{\bm x}{\bm x}=|d_{36}|$ with its corresponding azimuth angle $\varphi^*$ and  phase angle $\theta^*$,  $$\varphi^*=\dfrac{3\pi}4\mbox{ or }\dfrac{7\pi}4;\ \  \theta^*=\begin{cases}
 \frac{\pi}2,   d_{36}>0, \\
 \frac{3\pi}2, d_{36}<0.
 \end{cases}$$
 
 $$ \theta=0, \pi,\frac\pi2, \frac{3\pi}2\Longrightarrow {\bm \chi}_{\text{eff}}^{ee-o}{\bm x}{\bm y}{\bm y}=0\mbox{ or } \varphi=\frac{\pi}4, \frac{3\pi}4, \frac{5\pi}4, \frac{7\pi}4 \Longrightarrow {\bm \chi}_{\text{eff}}^{ee-o}{\bm x}{\bm y}{\bm y}=0. $$  

${\bm \chi}^{\text{eff}}_2=\max{\bm \chi}_{\text{eff}}^{ee-o}{\bm x}{\bm y}{\bm y}= |d_{14}|$ with its corresponding azimuth angle $\varphi^*$ and  phase angle $\theta^*$, $$d_{14}>0 \Longrightarrow \varphi^*=0\mbox{ or }\pi\mbox{ and } \theta^*=\dfrac{\pi}4\mbox{ or }\dfrac{5\pi}4; \varphi^*=\dfrac{\pi}2\mbox{ or }\dfrac{3\pi}2\mbox{ and } \theta^*= \dfrac{3\pi}4\mbox{ or }\dfrac{7\pi}4,$$
$$d_{14}<0 \Longrightarrow \varphi^*=\dfrac{\pi}2\mbox{ or }\dfrac{3\pi}2\mbox{ and } \theta^*=\dfrac{\pi}4\mbox{ or }\dfrac{5\pi}4; \varphi^*=0\mbox{ or }\pi\mbox{ and } \theta^*= \dfrac{3\pi}4\mbox{ or }\dfrac{7\pi}4.$$  

\textbf{3. Point Group 3m (e.g., LiNbO$_3$, LiTaO$_3$) }
$$\tan \theta=-\frac{d_{22}\sin3\varphi}{d_{31}} \Longrightarrow {\bm \chi}_{\text{eff}}^{oo-e}{\bm y}{\bm x}{\bm x}=0.$$ 

${\bm \chi}_1^{\textbf{eff}}=\max{\bm \chi}_{\text{eff}}^{oo-e}{\bm y}{\bm x}{\bm x}=\sqrt{d^2_{22}+d_{31}^2}$ with its corresponding azimuth angle $\varphi^*$ and  phase angle $\theta^*$, $$\varphi^*=\dfrac{\pi}6\mbox{ or }\dfrac{5\pi}6\mbox{ or }\dfrac{3\pi}2\mbox{ or }\dfrac{\pi}2\mbox{ or }\dfrac{7\pi}6\mbox{ or }\dfrac{11\pi}6;$$
$$ \theta^*=\arccos\frac{d_{22}}{\sqrt{d^2_{22}+d_{31}^2}}\mbox{ or }\pi-\arccos\frac{d_{22}}{\sqrt{d^2_{22}+d_{31}^2}}.$$

$$\theta=\frac\pi2, \frac{3\pi}2\Longrightarrow {\bm \chi}_{\text{eff}}^{ee-o}{\bm x}{\bm y}{\bm y}=0\mbox{ or } \varphi=\frac\pi6, \frac{\pi}2, \frac{3\pi}2,  \frac{5\pi}6, \frac{7\pi}6,\frac{11\pi}6  \Longrightarrow {\bm \chi}_{\text{eff}}^{ee-o}{\bm x}{\bm y}{\bm y}=0.$$  

${\bm \chi}^{\text{eff}}_2=\max{\bm \chi}_{\text{eff}}^{ee-o}{\bm x}{\bm y}{\bm y}= |d_{22}|$ with its corresponding azimuth angle $\varphi^*$ and  phase angle $\theta^*$,  $$d_{22}>0 \Longrightarrow \varphi^*=\dfrac{\pi}6\mbox{ or }\dfrac{5\pi}6\mbox{ or }\dfrac{3\pi}2\mbox{ and }\theta^*=0;$$
$$d_{22}<0 \Longrightarrow \varphi^*=\dfrac{\pi}2\mbox{ or }\dfrac{7\pi}6\mbox{ or }\dfrac{11\pi}6\mbox{ and } \theta^*=0.$$

\textbf{4. Point Group $\bar{6}$ (e.g., Ag$_2$HPO$_4$)}
 $$ \theta=\frac\pi2,\frac{3\pi}2 \Longrightarrow {\bm \chi}_{\text{eff}}^{oo-e}{\bm y}{\bm x}{\bm x}=0$$  or $$\varphi=\frac13\arctan\frac{d_{11}}{d_{22}}, \frac{2\pi}3+\frac13\arctan\frac{d_{11}}{d_{22}}, \frac{4\pi}3+ \frac13\arctan\frac{d_{11}}{d_{22}}  \Longrightarrow {\bm \chi}_{\text{eff}}^{oo-e}{\bm y}{\bm x}{\bm x}=0.$$
 
 ${\bm \chi}_1^{\textbf{eff}}=\max{\bm \chi}_{\text{eff}}^{oo-e}{\bm y}{\bm x}{\bm x}=\sqrt{d^2_{22}+d_{11}^2}$ with its corresponding azimuth angle $\varphi^*$ and  phase angle $\theta^*$, $$\varphi^*=\alpha, \alpha+\dfrac{2\pi}3, \alpha+\dfrac{4\pi}3,\ \theta^*=0
 \mbox{ or }\varphi^*=\alpha^*, \alpha^*+\dfrac{2\pi}3, \alpha^*+\dfrac{4\pi}3,\ \theta^*=\pi,$$ where $\alpha=\dfrac{2\pi}3-\dfrac13\arcsin\dfrac{d_{22}}{\sqrt{d_{11}^{2}+d_{22}^{2}}}$ and $\alpha^*=\dfrac{\pi}3-\dfrac13\arcsin\dfrac{d_{22}}{\sqrt{d_{11}^{2}+d_{22}^{2}}}.$\\
 
 $$\theta=\frac\pi2, \frac{3\pi}2\Longrightarrow {\bm \chi}_{\text{eff}}^{ee-o}{\bm x}{\bm y}{\bm y}=0$$  or  $$ \varphi=-\frac13\arctan\frac{d_{22}}{d_{11}}, \frac{2\pi}3 -\frac13\arctan\frac{d_{22}}{d_{11}}, \frac{4\pi}3  -\frac13\arctan\frac{d_{22}}{d_{11}}\Longrightarrow {\bm \chi}_{\text{eff}}^{ee-o}{\bm x}{\bm y}{\bm y}=0.$$

${\bm \chi}^{\text{eff}}_2=\max{\bm \chi}_{\text{eff}}^{ee-o}{\bm x}{\bm y}{\bm y}= \sqrt{d_{11}^{2}+d_{22}^{2}}$ with its corresponding azimuth angle $\varphi^*$ and  phase angle $\theta^*$, 
$$\varphi^*=\beta, \beta+\dfrac{2\pi}3, \beta+\dfrac{4\pi}3,\ \theta^*=0, \beta=\dfrac13\arcsin\frac{d_{11}}{\sqrt{d_{11}^{2}+d_{22}^{2}}}.$$

\textbf{5. Point Groups 6 and 4 (e.g., LiIO$_3$, $\alpha$-LiIO$_3$) }$$ \theta=0,\pi \Longrightarrow {\bm \chi}_{\text{eff}}^{oo-e}{\bm y}{\bm x}{\bm x}=0.$$

${\bm \chi}_1^{\textbf{eff}}=\max{\bm \chi}_{\text{eff}}^{oo-e}{\bm y}{\bm x}{\bm x}=|d_{31}|$ with its corresponding azimuth angle $\varphi^*$ and  phase angle $\theta^*$, $$\varphi^*\in[0,\pi],  \theta^*=\begin{cases}
	\frac{\pi}2, &d_{31}>0,\\
	\frac{3\pi}2, &d_{31}<0.
\end{cases}$$

$$\theta=0,\pi,\frac\pi2, \frac{3\pi}2\Longrightarrow {\bm \chi}_{\text{eff}}^{ee-o}{\bm x}{\bm y}{\bm y}=0.$$ 

${\bm \chi}^{\text{eff}}_2=\max{\bm \chi}_{\text{eff}}^{ee-o}{\bm x}{\bm y}{\bm y}= |d_{14}|$ with its corresponding azimuth angle $\varphi^*$ and  phase angle $\theta^*$,$$d_{14}>0 \Longrightarrow \varphi^*\in[0,\pi] \mbox{ and }\theta^*=\dfrac{\pi}4\mbox{ or }\dfrac{5\pi}4 ;$$
$$d_{14}<0 \Longrightarrow \varphi^*\in[0,\pi] \mbox{ and } \theta^*= \dfrac{3\pi}4\mbox{ or }\dfrac{7\pi}4.$$

\textbf{6. Point Group 3 (e.g., Na$_2$I$_2$O$_3$-6H$_2$O)} $$\tan\theta=\frac{d_{22}\sin3\varphi-d_{11}\cos3\varphi}{d_{31}}\Longrightarrow {\bm \chi}_{\text{eff}}^{oo-e}{\bm y}{\bm x}{\bm x}=0.$$ 

${\bm \chi}_1^{\textbf{eff}}=\max{\bm \chi}_{\text{eff}}^{oo-e}{\bm y}{\bm x}{\bm x}=\sqrt{d_{11}^2+d_{22}^2+d_{31}^2}$ with its corresponding azimuth angle $\varphi^*$ and  phase angle $\theta^*$, $\beta=\frac13\arcsin \frac{d_{22}}{\sqrt{d_{11}^2+d_{22}^2}},$ $$\varphi^*=\frac\pi3-\beta, \pi-\beta,\frac{5\pi}3-\beta, \theta^*=\arcsin \frac{d_{
		31}}{\sqrt{d_{11}^2+d_{22}^2+d_{31}^2}},$$  or  $$\varphi^*=\frac{2\pi}3-\beta, \frac{4\pi}3-\beta,2\pi-\beta, \theta^*=\pi-\arcsin \frac{d_{
		31}}{\sqrt{d_{11}^2+d_{22}^2+d_{31}^2}}.$$

$$\tan\theta=-\frac{d_{11}\sin3\varphi+d_{22}\cos3\varphi}{2d_{14}}\Longrightarrow {\bm \chi}_{\text{eff}}^{ee-o}{\bm x}{\bm y}{\bm y}=0.$$ 

 ${\bm \chi}^{\text{eff}}_2=\max{\bm \chi}_{\text{eff}}^{ee-o}{\bm x}{\bm y}{\bm y}= \frac12\left(\sqrt{d_{11}^2+d_{22}^2}+\sqrt{d_{11}^2+d_{22}^2+4d_{14}^2}\right)$ with its corresponding azimuth angle $\varphi^*$ and  phase angle $\theta^*$,$$\varphi^*=\phi\mbox{ or }\phi+\frac{2\pi}3\mbox{ or } \phi+\frac{4\pi}3, \phi=\frac13\arcsin\frac{d_{11}}{\sqrt{d_{11}^2+d_{22}^2}}$$ and $\alpha=\frac{\sqrt{d_{11}^2+d_{22}^2}}{\sqrt{d_{11}^2+d_{22}^2+4d_{14}^2}},$
$$d_{14}>0 \Longrightarrow \theta^*=\pi-\arccos\sqrt{\frac{1+\alpha}2}\mbox{ or }2\pi-\arccos\sqrt{\frac{1+\alpha}2};$$
$$d_{14}<0 \Longrightarrow  \theta^*= \arccos\sqrt{\frac{1+\alpha}2}\mbox{ or }\pi+\arccos\sqrt{\frac{1+\alpha}2}.$$ 

\textbf{7. Point Groups 4mm and 6mm (e.g., BaTiO$_3$, ZnO, C$_4$H$_4$O$_2$N)} $$ \theta=0,\pi \Longrightarrow {\bm \chi}_{\text{eff}}^{oo-e}{\bm y}{\bm x}{\bm x}=0.$$ 

${\bm \chi}_1^{\textbf{eff}}=\max{\bm \chi}_{\text{eff}}^{oo-e}{\bm y}{\bm x}{\bm x}=|d_{31}|$ with its corresponding azimuth angle $\varphi^*$ and  phase angle $\theta^*$,$$\varphi^*\in[0,2\pi],\mbox{ and } \theta^*=\frac\pi2 \ \ (d_{31}>0),\frac{3\pi}2  \ \  (d_{31}<0).$$  
$${\bm \chi}_{\text{eff}}^{ee-o}{\bm x}{\bm y}{\bm y}=0.$$

\textbf{8. Point Group 32 (e.g., $\alpha$-SiO$_2$, HgO) }$$ \theta=\frac\pi2, \frac{3\pi}2\Longrightarrow {\bm \chi}_{\text{eff}}^{oo-e}{\bm y}{\bm x}{\bm x}=0$$
or $$\varphi= \frac\pi6, \frac{\pi}2, \frac{3\pi}2,  \frac{5\pi}6, \frac{7\pi}6,\frac{11\pi}6  \Longrightarrow {\bm \chi}_{\text{eff}}^{oo-e}{\bm y}{\bm x}{\bm x}=0.$$
${\bm \chi}_1^{\textbf{eff}}=\max{\bm \chi}_{\text{eff}}^{oo-e}{\bm y}{\bm x}{\bm x}=|d_{11}|$ with its corresponding azimuth angle $\varphi^*$ and  phase angle $\theta^*$,$$d_{11}>0 \Longrightarrow \varphi^*=0\mbox{ or }\frac{2\pi}3\mbox{ or }\frac{4\pi}3, \theta^*=0;  \ \ \varphi^*=\frac\pi3\mbox{ or }\pi\mbox{ or }\frac{5\pi}3,\theta^*=\pi;$$  $$d_{11}<0 \Longrightarrow \varphi^*=\pi\mbox{ or }\frac{2\pi}3\mbox{ or }\frac{4\pi}3, \theta^*=\pi; \ \ \varphi^*=\frac\pi3\mbox{ or }\pi\mbox{ or }\frac{5\pi}3, \theta^*=0.$$

 $$\tan\theta=\frac{d_{11}\sin3\varphi}{2d_{14}}\Longrightarrow {\bm \chi}_{\text{eff}}^{ee-o}{\bm x}{\bm y}{\bm y}=0.$$ 

${\bm \chi}^{\text{eff}}_2=\max{\bm \chi}_{\text{eff}}^{ee-o}{\bm x}{\bm y}{\bm y}= \frac12\left(|d_{11}|+\sqrt{d^2_{11}+4d^2_{14}}\right)$ with its corresponding azimuth angle $\varphi^*$ and  phase angle $\theta^*$, $$\varphi^*=\frac{\pi}6 \mbox{ or }\frac{5\pi}6\mbox{ or } \frac{3\pi}2(d_{11}>0);   \varphi^*=\frac{\pi}2 \mbox{ or }\frac{7\pi}6 \mbox{ or } \frac{11\pi}6\ (d_{11}<0),$$
$$d_{14}>0 \Longrightarrow \theta^*=\pi-\arccos\sqrt{\frac{1+\alpha}2}\mbox{ or }2\pi-\arccos\sqrt{\frac{1+\alpha}2};$$
$$d_{14}<0 \Longrightarrow  \theta^*= \arccos\sqrt{\frac{1+\alpha}2}\mbox{ or }\pi+\arccos\sqrt{\frac{1+\alpha}2}.$$

\textbf{9. Point Groups 422 and 622  (e.g., Cl$_3$CCO$_2$K,  SiO$_2$)}\\

${\bm \chi}_{\text{eff}}^{oo-e}{\bm y}{\bm x}{\bm x}=0$, $$\theta=0,\pi,\frac\pi2, \frac{3\pi}2\Longrightarrow {\bm \chi}_{\text{eff}}^{ee-o}{\bm x}{\bm y}{\bm y}=0.$$ 

${\bm \chi}^{\text{eff}}_2=\max{\bm \chi}_{\text{eff}}^{ee-o}{\bm x}{\bm y}{\bm y}= |d_{14}|$ with its corresponding azimuth angle $\varphi^*$ and  phase angle $\theta^*$,   $$\varphi^*\in[0,2\pi], \theta^*=\frac{\pi}4 \mbox{ or }\frac{5\pi}4\  (d_{11}<0);  \frac{3\pi}4 \mbox{ or } \frac{7\pi}4 \ (d_{11}>0).$$

\textbf{10. Point Groups $\bar{6}$m2  (e.g., BaTiSi$_3$O$_9$)}$$\theta=\frac\pi2, \frac{3\pi}2\Longrightarrow {\bm \chi}_{\text{eff}}^{oo-e}{\bm y}{\bm x}{\bm x}=0$$or $$\varphi=0, \frac\pi3, \frac{2\pi}3, \pi, \frac{4\pi}3, \frac{5\pi}3 \Longrightarrow {\bm \chi}_{\text{eff}}^{oo-e}{\bm y}{\bm x}{\bm x}=0.$$ 

${\bm \chi}_1^{\textbf{eff}}=\max{\bm \chi}_{\text{eff}}^{oo-e}{\bm y}{\bm x}{\bm x}=|d_{11}|$ with its corresponding azimuth angle $\varphi^*$ and  phase angle $\theta^*$,$$\varphi^*=\frac{\pi}6 \mbox{ or }\frac{5\pi}6 \mbox{ or }\frac{3\pi}2, \theta^*=0\ (d_{22}<0), \pi\ (d_{22}>0);$$$$\varphi^*= \frac{\pi}2\mbox{ or } \frac{7\pi}6 \mbox{ or }\frac{11\pi}6, \theta^*=0\ (d_{22}>0), \pi\ (d_{22}<0).$$

 $$\theta=\frac\pi2, \frac{3\pi}2\Longrightarrow {\bm \chi}_{\text{eff}}^{ee-o}{\bm x}{\bm y}{\bm y}=0$$  or $$\varphi=\frac{\pi}6, \frac{\pi}2, \frac{3\pi}2, \frac{5\pi}6, \frac{7\pi}6, \frac{11\pi}6  \Longrightarrow {\bm \chi}_{\text{eff}}^{ee-o}{\bm x}{\bm y}{\bm y}=0.$$

${\bm \chi}^{\text{eff}}_2=\max{\bm \chi}_{\text{eff}}^{ee-o}{\bm x}{\bm y}{\bm y}= |d_{22}|$ with its corresponding azimuth angle $\varphi^*$ and  phase angle $\theta^*$,   $$d_{22}>0\Longrightarrow\varphi^*=0\mbox{ or }\frac{2\pi}3  \mbox{ or }\frac{4\pi}3, \theta^*=0\mbox{ or }\pi;$$$$d_{22}<0\Longrightarrow\varphi^*=\frac\pi3\mbox{ or }\pi  \mbox{ or }\frac{5\pi}3, \theta^*=0\mbox{ or }\pi.$$

%\section{Conclusion} In this paper, the calculation formulae of phase angle are given for two classes of largest effective SHG coefficients in uniaxial crystals by means of  the optimization theory.  With the help of such calculation formulae, we present the best phase angles and azimuth angles  of all  uniaxial crystals class,   as well as their effective SHG coefficients. Furthermore, these calculation  is only dependent upon some principal subtensor of second order susceptibility tensor. 
\section*{Declaration of competing interest}
The author declare that they have no known competing financial interests or personal rela-tionships that could have appeared to influence the work reported in this paper.
\section*{Authors' contributions}Yisheng Song is the sole author.
\section*{Availability of data and materials}
This manuscript has no associated data or the data will not be deposited. 
\section*{Funding}
This  work was supported by
the National Natural Science Foundation of P.R. China (Grant No.12171064).
\section{Appendix A}

{\bf Proof of Lemma \ref{lem:21}.}  It is well-known that $$\max\{a_{11}x_1^2+2a_{12}x_1x_2+a_{22}x_2^2: x_1^2+x_2^2=1\}=\max\{\lambda: {\bm A}{\bm x}=\lambda{\bm x}\},$$
where ${\bm A}=\begin{pmatrix}
	a_{11}   \  a_{12}\\
	a_{12} \  a_{22}
\end{pmatrix}$ is the  coefficient matrix of the above quadratic form. That is, such a maximum  is  the largest eigenvalue of $ {\bm A}$.  Then  the  roots of characteristic polynomial of $ {\bm A}$ may be solved as follows,
$$(a_{11}-\lambda ) (a_{22}-\lambda ) -a^2_{12}=\lambda^2-(a_{11}+a_{22})\lambda+a_{11}a_{22}-a^2_{12}=0,$$
and so, $$\lambda=\dfrac{a_{11}+a_{22}\pm\sqrt{(a_{11}-a_{22})^2+4a^2_{12}}}2.$$
Thus, the desired maximum is $\lambda_1=\frac12(a_{11}+a_{22}+\sqrt{(a_{11}-a_{22})^2+4a^2_{12}}), $ and the corresponding  eigenvectors may be the roots of the following equations systems $$\begin{cases}
	(a_{11}-\lambda_1)x_1+a_{12}x_2=0,\\
	(a_{22}-\lambda_1)x_2+a_{12}x_1=0,
\end{cases}$$
and then whenever $a_{12}=0$,  we have ${\bm x}=\begin{cases}
	(1,0)^\top, &a_{11}>a_{22},\\
	(x_1,x_2)^\top, &a_{11}=a_{22},\\
	(0, 1)^\top, &a_{11}<a_{22}.
\end{cases}$  

If $a_{12}\ne0$,  then $${\bm x}=\begin{pmatrix}x_1\\x_2\end{pmatrix}=\begin{pmatrix}a_{12}\\\lambda_1-a_{11}\end{pmatrix}\mbox{ or }\begin{pmatrix}\lambda_1-a_{22}\\a_{12}\end{pmatrix}.$$
Therefore,  the maximizer is  $\dfrac{\pm{\bm x}}{\|{\bm x}\|_2}$.  Noting $(a_{11}-\lambda_1 ) (a_{22}-\lambda_1 )=a^2_{12}, $ we have
$$ \|{\bm x}\|_2^2=a_{12}^2+(\lambda_1-a_{11})^2=(\lambda_1-a_{11})(2\lambda_1-a_{11}-a_{22})$$ or 
$$ \|{\bm x}\|_2^2=a_{12}^2+(\lambda_1-a_{22})^2=(\lambda_1-a_{22})(2\lambda_1-a_{11}-a_{22}),$$
and $$a_{12}=\begin{cases}
	\sqrt{(\lambda_1 -a_{11}) (\lambda_1-a_{22})}, & a_{12}>0,\\
	-\sqrt{(\lambda_1 -a_{11}) (\lambda_1-a_{22})}, & a_{12}<0.
\end{cases}$$
By this time,  it is easy to yield $$\lambda_1 -a_{11}=\dfrac12\left(a_{22}-a_{11}+\sqrt{(a_{11}-a_{22})^2+4a^2_{12}}\right)\ge0$$ and $\lambda_1 -a_{22}=\dfrac12\left(a_{11}-a_{22}+\sqrt{(a_{11}-a_{22})^2+4a^2_{12}}\right)\ge0$.
So, we obtain $$\dfrac{{\bm x}}{\|{\bm x}\|_2}=\begin{cases}
	\dfrac{\left(\sqrt{\lambda_1 -a_{22}}, \sqrt{\lambda_1-a_{11}}\right)^\top}{\sqrt{2\lambda_1 -a_{11}-a_{22}}}, & a_{12}\ge0,\\
	(y_1,y_2)^\top \mbox{ with }y_1^2+y_2^2=1, & a_{12}=0\mbox{ and }a_{11}=a_{22},\\
	\dfrac{\left(-\sqrt{\lambda_1 -a_{22}}, \sqrt{\lambda_1-a_{11}}\right)^\top}{\sqrt{2\lambda_1 -a_{11}-a_{22}}}, & a_{12}\le0.
\end{cases}$$
By the simple calculation, the desired conclusions follow.  \qed

\section{Appendix B}

\begin{lemma}[Trigonometric Form of Cardano's Formula]\label{lem:51}
	Let a cubic equation be the following form, $t^3+pt+q=0.$   Assume  $$\Delta=\left(\frac{q}{2}\right)^2+\left(\frac{p}3\right)^3, \cos3\theta=\frac{-\frac{q}{2}}{\sqrt{\left(-\frac{p}3\right)^3}}=\frac{3\sqrt3q}{2p\sqrt{-p}}.$$  If the discriminant  $\bigtriangleup<0$,  then solve such a cubic equation to yield $$\begin{cases}
		t_1=2\sqrt{-\dfrac{p}3}\cos\theta,\\
		t_2=2\sqrt{-\dfrac{p}3}\cos\left(\theta+\dfrac{2\pi}3\right),\\
		t_3=2\sqrt{-\dfrac{p}3}\cos\left(\theta+\dfrac{4\pi}3\right).
	\end{cases}$$
\end{lemma}
\begin{example}[Point Group $\bar{4}$ (e.g., HgGa$_2$S$_4$)] \label{ex:1}  In Point Group  $\bar{4}$ Crystal, there are four non-zero independent components of second order susceptibility tensor ${\bm\chi}^{(2)}=(t_{ijk})$, $t_{311}=-t_{322}=d_{31}, t_{113}=t_{131}=-t_{223}=-t_{232}=d_{15}, $
	$t_{321}=t_{312}=d_{36}, t_{123}=t_{132}=t_{213}=t_{231}=d_{14} . $
\end{example}

Obviously, by Eqs. \eqref{eq:34} and \eqref{eq:38}, we have $$a_{11}({\bm x})=a_{22}({\bm x})=0, a_{12}({\bm x})=2d_{15}x_1x_2+d_{14}(x_2^2-x_1^2),$$$$V_1({\bm x})=V_2({\bm x})=0, V_3({\bm x})=2d_{36}x_1x_2+d_{31}(x_1^2-x_2^2).$$
Then  from Eqs. \eqref{eq:37} and \eqref{eq:311},  it follows that  $$ \theta=0, \pi \Longrightarrow {\bm \chi}_{\text{eff}}^{oo-e}{\bm y}{\bm x}{\bm x}=0;  \theta=0, \pi,\frac\pi2, \frac{3\pi}2 \Longrightarrow {\bm \chi}_{\text{eff}}^{ee-o}{\bm x}{\bm y}{\bm y}=0.$$ It is easy to show that $$ V_3({\bm x})=0, \mbox{i.e.,} \tan2\varphi=-\frac{d_{31}}{d_{36}}  \Longrightarrow {\bm \chi}_{\text{eff}}^{oo-e}{\bm y}{\bm x}{\bm x}=0;$$ $$ a_{12}({\bm x})=0, \mbox{i.e.,} \tan2\varphi=\frac{d_{14}}{d_{15}} \Longrightarrow {\bm \chi}_{\text{eff}}^{ee-o}{\bm x}{\bm y}{\bm y}=0.$$

It follows from Eqs. \eqref{eq:35} and \eqref{eq:36} that
$$\aligned {\bm \chi}_1^{\textbf{eff}}=&\max\limits_{\|{\bm x}\|_2=1}\sqrt{(V_3({\bm x}))^2} =\max\limits_{\varphi}|-d_{36}\sin2\varphi-d_{31}\cos2\varphi|\\=&\sqrt{d^2_{36}+d_{31}^2}\endaligned$$ with its maximizers ${\bm x}^*=(\sin\varphi^*, -\cos\varphi^*,0)$ with $ \sin2\varphi^*=\pm\frac{d_{36}}{\sqrt{d^2_{36}+d_{31}^2}}, \cos2\varphi^*=\pm\frac{d_{31}}{\sqrt{d^2_{36}+d_{31}^2}},$ ${\bm z}^*= (\cos\theta^*,\sin\theta^*)=\mp	(0,1),$
and  hence, its corresponding azimuth angle $\varphi^*$ and  phase angle $\theta^*$,  $$\varphi^*=\frac12\arcsin \frac{d_{36}}{\sqrt{d^2_{36}+d_{31}^2}}, \theta^*=\frac{3\pi}2\mbox{ or }\varphi^*=\frac{\pi}2+\frac12\arcsin \frac{d_{36}}{\sqrt{d^2_{36}+d_{31}^2}},  \theta^*=\frac{\pi}2 .$$ 

It follows from Eqs. \eqref{eq:39} and \eqref{eq:310} that
$${\bm \chi}_2^{\textbf{eff}}=\max\limits_{\|{\bm x}\|_2=1}\frac12\left(0+\sqrt{4(a_{12}({\bm x}))^2}\right)=\max\limits_{\varphi}|d_{14}\cos2\varphi-d_{15}\sin2\varphi|=\sqrt{d^2_{15}+d_{14}^2}$$ with its maximizer $${\bm x}^*=(\sin\varphi^*, -\cos\varphi^*,0)\mbox{ with } \sin2\varphi^*=\pm\frac{d_{15}}{\sqrt{d^2_{14}+d_{15}^2}}, \cos2\varphi^*=\mp\frac{d_{14}}{\sqrt{d^2_{14}+d_{15}^2}},$$ $${\bm z}^*=(\cos\theta^*, \sin\theta^*)=\begin{cases}
\pm\left(\frac{\sqrt2}2,\frac{\sqrt2}2\right),\ \  a_{12}({\bm x}^*)>0, \\
\pm\left(\frac{\sqrt2}2,-\frac{\sqrt2}2\right), a_{12}({\bm x}^*)<0.
\end{cases}
$$ Then, its corresponding azimuth angle $\varphi^*$ and  phase angle $\theta^*$, $$\varphi^*=\frac\pi2-\frac12\arcsin \frac{d_{15}}{\sqrt{d^2_{14}+d_{15}^2}}, \theta^*= \frac{3\pi}4\mbox{ or }\frac{7\pi}4;$$$$\varphi^*=\pi-\frac12\arcsin \frac{d_{15}}{\sqrt{d^2_{14}+d_{15}^2}}, \theta^*= \frac{\pi}4\mbox{ or }\frac{5\pi}4.$$

\begin{example}[Point Group $\bar{4}2$m (e.g., KH$_2$PO$_4$)] \label{ex:2}   In Point Group  $\bar{4}2$m Crystal, there are two non-zero independent components of second order susceptibility tensor ${\bm\chi}^{(2)}=(t_{ijk})$,  
	$t_{321}=t_{312}=d_{36}, t_{123}=t_{132}=t_{213}=t_{231}=d_{14} . $
\end{example}

Obviously, by Eqs. \eqref{eq:34} and \eqref{eq:38}, we have $$a_{11}({\bm x})=a_{22}({\bm x})=0, a_{12}({\bm x})=d_{14}(x_2^2-x_1^2),$$    $$V_1({\bm x})=V_2({\bm x})=0, V_3({\bm x})=2d_{36}x_1x_2.$$
Then from Eqs. \eqref{eq:37} and \eqref{eq:311},  it follows that  $$ \theta=0, \pi \Longrightarrow {\bm \chi}_{\text{eff}}^{oo-e}{\bm y}{\bm x}{\bm x}=0; \  \theta=0, \pi,\frac\pi2, \frac{3\pi}2\Longrightarrow {\bm \chi}_{\text{eff}}^{ee-o}{\bm x}{\bm y}{\bm y}=0. $$ It is easy to show that $$V_3({\bm x})=-d_{36}\sin2\varphi=0, \mbox{i.e., } \varphi=0, \frac\pi2, \frac{3\pi}2, \pi  \Longrightarrow {\bm \chi}_{\text{eff}}^{oo-e}{\bm y}{\bm x}{\bm x}=0, $$ $$ a_{12}({\bm x})=d_{14}\cos2\varphi=0, \mbox{i.e., } \varphi=\frac{\pi}4, \frac{3\pi}4, \frac{5\pi}4, \frac{7\pi}4 \Longrightarrow {\bm \chi}_{\text{eff}}^{ee-o}{\bm x}{\bm y}{\bm y}=0.$$

It follows from Eqs. \eqref{eq:35} and \eqref{eq:36} that $${\bm \chi}_1^{\textbf{eff}}=\max\limits_{\|{\bm x}\|_2=1}\sqrt{0+(V_3({\bm x}))^2}\\=\max\limits_{\|{\bm x}\|_2=1}|2d_{36}x_1x_2|=|d_{36}|$$ with its maximizer $${\bm x}^*=\pm\left(\frac{\sqrt2}2,\frac{\sqrt2}2,0\right),  {\bm z}^*=\begin{cases}
	(0,1), d_{36}>0,\\
	(0,-1), d_{36}<0.
\end{cases}$$ Hence, its corresponding azimuth angle $\varphi^*$ and  phase angle $\theta^*$, $$\varphi^*=\frac{3\pi}4\mbox{ or }\frac{7\pi}4,  \theta^*=\frac{\pi}2 (d_{36}>0), \frac{3\pi}2  (d_{36}<0).$$

It follows from Eqs. \eqref{eq:39} and \eqref{eq:310} that $${\bm \chi}_2^{\textbf{eff}}=\frac12\max\limits_{\|{\bm x}\|_2=1}(0+\sqrt{0+4(a_{12}({\bm x}))^2})=\max\limits_{\|{\bm x}\|_2=1}|d_{14}(x_2^2-x_1^2)|= |d_{14}|$$ with its maximizer ${\bm x}^*=(0,\pm1,0), \mbox{ or }
(\pm1,0,0);
$
$${\bm z}^*=\begin{cases}
	\pm\left(\frac{\sqrt2}2,\frac{\sqrt2}2\right), {\bm x}^*=(0,\pm1,0)\mbox{ and }d_{14}>0,\\
	\pm\left(\frac{\sqrt2}2,-\frac{\sqrt2}2\right), {\bm x}^*=(0,\pm1,0)\mbox{ and }d_{14}<0;\\
	\pm\left(\frac{\sqrt2}2,\frac{\sqrt2}2\right), {\bm x}^*=(\pm1,0,0)\mbox{ and }d_{14}<0,\\
	\pm\left(\frac{\sqrt2}2,-\frac{\sqrt2}2\right), {\bm x}^*=(\pm1,0,0)\mbox{ and }d_{14}>0;
\end{cases}$$  and  hence, its corresponding azimuth angle $\varphi^*$ and  phase angle $\theta^*$,      $$d_{14}>0 \Longrightarrow \begin{cases}\varphi^*=0\mbox{ or }\pi\mbox{ and } \theta^*=\frac{\pi}4\mbox{ or }\frac{5\pi}4; \\ \varphi^*=\frac{\pi}2\mbox{ or }\frac{3\pi}2\mbox{ and } \theta^*= \frac{3\pi}4\mbox{ or }\frac{7\pi}4.\end{cases}$$
$$d_{14}<0 \Longrightarrow \begin{cases}
	\varphi^*=\frac{\pi}2\mbox{ or }\frac{3\pi}2\mbox{ and } \theta^*=\frac{\pi}4\mbox{ or }\frac{5\pi}4;\\ \varphi^*=0\mbox{ or }\pi\mbox{ and } \theta^*= \frac{3\pi}4\mbox{ or }\frac{7\pi}4.
\end{cases}$$  

\begin{example}[Point Group 3m (e.g., LiNbO$_3$, LiTaO$_3$)]  \label{ex:3}  In Point Group 3mm Crystal, there are four non-zero independent components of second order susceptibility tensor ${\bm\chi}^{(2)}=(t_{ijk})$,  
	$t_{322}=t_{311}=d_{31}, t_{223}=t_{232}=t_{113}=t_{131}=d_{15}, t_{222}=-t_{211}=-t_{121}=-t_{112}=d_{22},  t_{333}=d_{33} .$
\end{example}

Obviously, by Eqs. \eqref{eq:34} and \eqref{eq:38}, we have $$a_{11}({\bm x})=d_{22}(3 x_1^2x_2-x_2^3), a_{22}({\bm x})=a_{12}({\bm x})=0,$$$$ V_1({\bm x})=-2d_{22}x_1x_2, V_2({\bm x})=d_{22}(x_2^2-x_1^2), V_3({\bm x})=d_{31}(x_1^2+x^2_2).$$
Then from Eqs. \eqref{eq:37} and \eqref{eq:311},  it follows that $$\tan \theta=-\frac{d_{22}\sin3\varphi}{d_{31}} \Longrightarrow {\bm \chi}_{\text{eff}}^{oo-e}{\bm y}{\bm x}{\bm x}=0;\ \theta=\frac\pi2, \frac{3\pi}2\Longrightarrow {\bm \chi}_{\text{eff}}^{ee-o}{\bm x}{\bm y}{\bm y}=0.$$ It is easy to show that $$ a_{11}({\bm x})=d_{22}\cos3\varphi=0, \mbox{i.e., } \varphi=\frac\pi6, \frac{\pi}2, \frac{3\pi}2,  \frac{5\pi}6, \frac{7\pi}6,\frac{11\pi}6  \Longrightarrow {\bm \chi}_{\text{eff}}^{ee-o}{\bm x}{\bm y}{\bm y}=0.$$

It follows from Eqs. \eqref{eq:35} and \eqref{eq:36} that   $${\bm \chi}_1^{\textbf{eff}}
=\max\limits_{\|{\bm x}\|_2=1}\sqrt{d^2_{22}(x_1^3-3x_1x_2^2)^2+d_{31}^2}=\sqrt{d^2_{22}+d_{31}^2}
$$ with its maximizer $${\bm x}^*=(\sin\varphi^*, -\cos\varphi^*,0)\mbox{ with }\sin^23\varphi^*=1;$$ $$ {\bm z}^*=(\cos\theta^*,\sin\theta^*)=\frac1{\sqrt{d^2_{22}+d_{31}^2}}(\pm d_{22}, d_{31}),$$ and  hence, its corresponding azimuth angle $\varphi^*$ and  phase angle $\theta^*$,  $$\varphi^*=\frac{\pi}6\mbox{ or }\frac{5\pi}6\mbox{ or }\frac{3\pi}2\mbox{ or }\frac{\pi}2\mbox{ or }\frac{7\pi}6\mbox{ or }\frac{11\pi}6;$$
$$ \theta^*=\arccos\frac{d_{22}}{\sqrt{d^2_{22}+d_{31}^2}}\mbox{ or }\pi-\arccos\frac{d_{22}}{\sqrt{d^2_{22}+d_{31}^2}}.$$

It follows from Eqs. \eqref{eq:39} and \eqref{eq:310} that $$\aligned {\bm \chi}_2^{\textbf{eff}}=&\frac12\max\limits_{\|{\bm x}\|_2=1}(a_{11}({\bm x})+|a_{11}({\bm x})|)=\max\limits_{\|{\bm x}\|_2=1}a_{11}({\bm x})=\max\limits_{\|{\bm x}\|_2=1}d_{22}(3 x_1^2x_2-x_2^3)\\=&\max\limits_{\varphi}d_{22}\cos3\varphi=|d_{22}|\endaligned$$ with its maximizer $${\bm x}^*=(\sin\varphi^* -\cos\varphi^*,0)\mbox{ with }\cos3\varphi^*=1 (d_{22}>0),
-1 (d_{22}<0),$$ $ {\bm z}^*=(\cos\theta^*,\sin\theta^*)=\pm(1,0),$ and  hence, its corresponding azimuth angle $\varphi^*$ and  phase angle $\theta^*$, $$d_{22}>0 \Longrightarrow \varphi^*=0\mbox{ or }\frac{2\pi}3\mbox{ or }\frac{4\pi}3\mbox{ and }\theta^*=0 \mbox{ or }\pi;$$
$$d_{22}<0 \Longrightarrow \varphi^*=\frac{\pi}3\mbox{ or }\pi\mbox{ or }\frac{5\pi}3\mbox{ and } \theta^*=0\mbox{ or }\pi.$$

\begin{example}[Point Group $\bar{6}$ (e.g., Ag$_2$HPO$_4$)]  \label{ex:4}  In Point Group $\bar{6}$ Crystal, there are two non-zero independent components of second order susceptibility tensor ${\bm\chi}^{(2)}=(t_{ijk})$, 
	$
	t_{111}=-t_{122}=-t_{212}=-t_{221}=d_{11}, 
	t_{222}=-t_{112}=-t_{121}=-t_{211}=d_{22}.
	$
\end{example}

Clearly, by Eqs. \eqref{eq:34} and \eqref{eq:38}, we have $$ 
V_1({\bm x})=d_{11}(x_{1}^{2}-x_{2}^{2})-2d_{22}x_{1}x_{2},  V_2({\bm x})=d_{22}(x_{2}^{2}-x_{1}^{2})-2d_{11}x_{1}x_{2}, V_3({\bm x})=0,
$$ 
$$   a_{11}({\bm x})=d_{11}(3x_{1}x_{2}^{2}-x_1^{3})+d_{22}(3x_{1}^{2}x_{2}-x_{2}^{3}), a_{12}({\bm x})=a_{22}({\bm x})=0.$$ 
Then from Eqs. \eqref{eq:37} and \eqref{eq:311},  it follows that $$ \theta=\frac\pi2,\frac{3\pi}2 \Longrightarrow {\bm \chi}_{\text{eff}}^{oo-e}{\bm y}{\bm x}{\bm x}=0; \theta=\frac\pi2, \frac{3\pi}2\Longrightarrow {\bm \chi}_{\text{eff}}^{ee-o}{\bm x}{\bm y}{\bm y}=0.$$ It is easy to show that $ x_2V_1({\bm x})-x_1V_2({\bm x})=d_{11}\cos3\varphi-d_{22}\sin3\varphi=0, $ i.e., $$\tan3\varphi=\frac{d_{11}}{d_{22}}  \Longrightarrow {\bm \chi}_{\text{eff}}^{oo-e}{\bm y}{\bm x}{\bm x}=0;$$  $ a_{11}({\bm x})=d_{11}\sin3\varphi+d_{22}\cos3\varphi=0, $ i.e., $$ \tan3\varphi=-\frac{d_{22}}{d_{11}} \Longrightarrow {\bm \chi}_{\text{eff}}^{ee-o}{\bm x}{\bm y}{\bm y}=0.$$

It follows from Eqs. \eqref{eq:35} and \eqref{eq:36} that  $$\aligned {\bm \chi}^{\text{eff}}_1=& \max\limits_{{\bm x}\|_2=1}\sqrt{(d_{11}(3x_{1}^{2}x_{2}-x_{2}^{3})-d_{22}(3x_{1}x_{2}^{2}-x_{1}^{3}))^{2}}\\
=& \max\limits_{{\bm x}\|_2=1}|d_{11}(3x_{1}^{2}x_{2}-x_{2}^{3})-d_{22}(3x_{1}x_{2}^{2}-x_{1}^{3})|
= \sqrt{d_{11}^{2}+d_{22}^{2}}\endaligned$$	
with its maximizer $ {\bm x}^*=(\sin\varphi^*, -\cos\varphi^*,0)$ with $$\cos3\varphi^*=\frac{\pm d_{11}}{\sqrt{d_{11}^{2}+d_{22}^{2}}},  \sin3\varphi=\frac{\mp d_{22}}{\sqrt{d_{11}^{2}+d_{22}^{2}}},\ {\bm z}^*=\pm
(1,0).$$ Let $\alpha=\frac{2\pi}3-\frac13\arcsin\frac{d_{22}}{\sqrt{d_{11}^{2}+d_{22}^{2}}}$ and $\alpha^*=\frac{\pi}3-\frac13\arcsin\frac{d_{22}}{\sqrt{d_{11}^{2}+d_{22}^{2}}}.$ Then its corresponding azimuth angle $\varphi^*$ and  phase angle $\theta^*$, 
$$\varphi^*=\alpha, \alpha+\frac{2\pi}3, \alpha+\frac{4\pi}3,\ \theta^*=0\mbox{
or }\varphi^*=\alpha^*, \alpha^*+\frac{2\pi}3, \alpha^*+\frac{4\pi}3,\ \theta^*=\pi.$$

It follows from Eqs. \eqref{eq:39} and \eqref{eq:310} that 
$$\aligned 
{\bm \chi}_{\text{eff}}^2=&\frac12\max\limits_{\|{\bm x}\|_2=1}\left\{a_{11}({\bm x})+|a_{11}({\bm x})|\right\}
=\max\limits_{\|{\bm x}\|_2=1}a_{11}({\bm x})\\ =&\max\limits_{\|{\bm x}\|_2=1}(d_{11}(3x_{1}x_{2}^{2}-x_1^{3})+d_{22}(3x_{1}^{2}x_{2}-x_{2}^{3}))=\sqrt{d_{11}^{2}+d_{22}^{2}}\endaligned$$ with its maximizer ${\bm x}^*=(\sin\varphi^*, -\cos\varphi^*,0)$ with $$ \cos3\varphi=\frac{d_{22}}{\sqrt{d_{11}^{2}+d_{22}^{2}}},  \sin3\varphi^*=\frac{d_{11}}{\sqrt{d_{11}^{2}+d_{22}^{2}}},\ {\bm z}^*=(\cos\theta^*,\sin\theta^*)=(1,0).$$
Let $\beta=\frac13\arcsin\frac{d_{11}}{\sqrt{d_{11}^{2}+d_{22}^{2}}}.$ Then its corresponding azimuth angle $\varphi^*$ and  phase angle $\theta^*$,  
$$\varphi^*=\beta, \beta+\frac{2\pi}3, \beta+\frac{4\pi}3,\ \theta^*=0.$$ 

\begin{example}[Point Groups 6 and 4 (e.g., LiIO$_3$, $\alpha$-LiIO$_3$)] \label{ex:5}   In Point Group 6 Crystal, there are four non-zero independent components of second order susceptibility tensor ${\bm\chi}^{(2)}=(t_{ijk})$, 
	$t_{322}=t_{311}=d_{31}, t_{123}=t_{132}=-t_{213}=-t_{231}=d_{14}, t_{223}=t_{232}=t_{113}=t_{131}=d_{15},  t_{333}=d_{33} .$
\end{example}

Obviously, by Eqs. \eqref{eq:34} and \eqref{eq:38}, we have $$a_{11}({\bm x})=a_{22}({\bm x})=0, a_{12}({\bm x})=-d_{14}(x_1^2+x_2^2),$$ $$V_1({\bm x})=V_2({\bm x})=0, V_3({\bm x})=d_{31}(x_1^2+x^2_2).$$
Then from Eqs. \eqref{eq:37} and \eqref{eq:311},  it follows that $$ \theta=0,\pi \Longrightarrow {\bm \chi}_{\text{eff}}^{oo-e}{\bm y}{\bm x}{\bm x}=0;$$   $$\theta=0,\pi,\frac\pi2, \frac{3\pi}2\Longrightarrow {\bm \chi}_{\text{eff}}^{ee-o}{\bm x}{\bm y}{\bm y}=0.$$ 

It follows from Eqs. \eqref{eq:35} and \eqref{eq:36} that  $$ {\bm \chi}_1^{\textbf{eff}}=\max\limits_{\|{\bm x}\|_2=1}\sqrt{0+(V_3({\bm x}))^2}
=\max\limits_{\|{\bm x}\|_2=1}|d_{31}(x_1^2+x^2_2)|=|d_{31}|
$$ with its maximizer $${\bm x}^*=(\sin\varphi^*, -\cos\varphi^*,0), \\ {\bm z}^*=
(0,1)\ (d_{31}>0),
(0,-1)\ (d_{31}<0);
$$ and  hence, its corresponding azimuth angle $\varphi^*$ and  phase angle $\theta^*$,  $$\varphi^*\in[0,2\pi],  \theta^*=
\frac{\pi}2\ (d_{31}>0), 
\frac{3\pi}2\ (d_{31}<0).
$$

It follows from Eqs. \eqref{eq:39} and \eqref{eq:310} that 
$${\bm \chi}_2^{\textbf{eff}}=\frac12\max\limits_{\|{\bm x}\|_2=1}(0+\sqrt{0+4(a_{12}({\bm x}))^2})=\max\limits_{\|{\bm x}\|_2=1}|a_{12}({\bm x})|=|d_{14}|$$ with its maximizer $${\bm x}^*=(\sin\varphi^*, -\cos\varphi^*,0),  {\bm z}^*=
\pm\left(\frac{\sqrt2}2,\frac{\sqrt2}2\right)\ (d_{14}<0),
\pm\left(\frac{\sqrt2}2,-\frac{\sqrt2}2\right)\ (d_{14}>0),
$$ and  hence, its corresponding azimuth angle $\varphi^*$ and  phase angle $\theta^*$,   $$d_{14}>0 \Longrightarrow \varphi^*\in[0,2\pi] \mbox{ and }\theta^*=\frac{\pi}4\mbox{ or }\frac{5\pi}4 ;$$ 
$$d_{14}<0 \Longrightarrow \varphi^*\in[0,2\pi] \mbox{ and } \theta^*= \frac{3\pi}4\mbox{ or }\frac{7\pi}4.$$ 

\begin{example}[Point Group 3 (e.g., Na$_2$I$_2$O$_3$-6H$_2$O)]\label{ex:6}    In Point Group 3 Crystal, there are six non-zero independent components of second order susceptibility tensor ${\bm\chi}^{(2)}=(t_{ijk})$, 
	$t_{111}=-t_{122}=-t_{212}=-t_{221}=d_{11}, t_{123}=t_{132}=-t_{213}=-t_{231}=d_{14},  t_{222}=-t_{211}=-t_{112}=-t_{121}=d_{22}, t_{223}=t_{232}=t_{113}=t_{131}=d_{15}, t_{322}=t_{311}=d_{31},  t_{333}=d_{33} .$
\end{example}

Obviously, by Eqs. \eqref{eq:34} and \eqref{eq:38}, we have $$a_{11}({\bm x})=d_{11}(3x_1x_2^2-x_1^3)+d_{22}(3x_1^2x_2-x_2^3), a_{22}({\bm x})=0,  a_{12}({\bm x})=-d_{14}(x_1^2+x_2^2),$$  $$V_1({\bm x})=d_{11}(x_1^2-x_2^2)-2d_{22}x_1x_2, V_2({\bm x})=d_{22}(x_2^2-x_1^2)-2d_{11}x_1x_2,  V_3({\bm x})=d_{31}(x_1^2+x^2_2).$$
Then from Eqs. \eqref{eq:37} and \eqref{eq:311},  it follows that $$\tan\theta=\frac{d_{22}\sin3\varphi-d_{11}\cos3\varphi}{d_{31}}\Longrightarrow {\bm \chi}_{\text{eff}}^{oo-e}{\bm y}{\bm x}{\bm x}=0;$$ $$\tan\theta=-\frac{d_{11}\sin3\varphi+d_{22}\cos3\varphi}{2d_{14}}\Longrightarrow {\bm \chi}_{\text{eff}}^{ee-o}{\bm x}{\bm y}{\bm y}=0.$$ 

It follows from Eqs. \eqref{eq:35} and \eqref{eq:36} that  $$\aligned {\bm \chi}_1^{\textbf{eff}}=&\max\limits_{\|{\bm x}\|_2=1}\sqrt{(d_{11}(x_2^3-3x_1^2x_2)+d_{22}(3x_1x_2^2-x_1^3))^2+(d_{31}(x_1^2+x^2_2))^2}
\\=& \max\limits_{\varphi}\sqrt{(d_{22}\sin3\varphi-d_{11}\cos3\varphi)^2+d_{31}^2}=\sqrt{d_{11}^2+d_{22}^2+d_{31}^2}\endaligned$$ with its maximizer $${\bm x}^*=(\sin\varphi^*, -\cos\varphi^*,0)\mbox{ with } \sin3\varphi^*=\pm\frac{d_{22}}{\sqrt{d_{11}^2+d_{22}^2}}, \cos3\varphi^*=\mp\frac{d_{11}}{\sqrt{d_{11}^2+d_{22}^2}},$$ $$ {\bm z}^*=(\cos\theta^*,\sin\theta^*)=\left(\pm\frac{\sqrt{d_{11}^2+d_{22}^2}}{\sqrt{d_{11}^2+d_{22}^2+d_{31}^2}}, \frac{d_{
		31}}{\sqrt{d_{11}^2+d_{22}^2+d_{31}^2}}\right)$$ and  hence, its corresponding azimuth angle   $\varphi^*$ and  phase angle $\theta^*$, $$\varphi^*=\frac\pi3-\beta, \pi-\beta,\frac{5\pi}3-\beta, \beta=\frac13\arcsin \frac{d_{22}}{\sqrt{d_{11}^2+d_{22}^2}},  \theta^*=\arcsin \frac{d_{
		31}}{\sqrt{d_{11}^2+d_{22}^2+d_{31}^2}},$$   $$\varphi^*=\frac{2\pi}3-\beta, \frac{4\pi}3-\beta,2\pi-\beta, \theta^*=\pi-\arcsin \frac{d_{
		31}}{\sqrt{d_{11}^2+d_{22}^2+d_{31}^2}}.$$

It follows from Eqs. \eqref{eq:39} and \eqref{eq:310} that $$\aligned {\bm \chi}_2^{\textbf{eff}}=&\frac12\max\limits_{\varphi}\left(d_{11}\sin3\varphi+d_{22}\cos3\varphi+\sqrt{(d_{11}\sin3\varphi+d_{22}\cos3\varphi)^2+4d_{14}^2}\right)\\
=&\frac12\left(\sqrt{d_{11}^2+d_{22}^2}+\sqrt{d_{11}^2+d_{22}^2+4d_{14}^2}\right)\endaligned$$ with its maximizer ${\bm x}^*=(\sin\varphi^*, -\cos\varphi^*,0)$ with $ \sin3\varphi^*=\frac{d_{11}}{\sqrt{d_{11}^2+d_{22}^2}},$ $${\bm z}^*=(\cos\theta^*,\sin\theta^*)=\begin{cases}
	\pm\left(\sqrt{\frac{1+\alpha}2},\sqrt{\frac{1-\alpha}2}\right), d_{14}<0,\\
	\pm\left(\sqrt{\frac{1+\alpha}2},-\sqrt{\frac{1-\alpha}2}\right), d_{14}>0,
\end{cases}$$ where $\alpha=\frac{\sqrt{d_{11}^2+d_{22}^2}}{\sqrt{d_{11}^2+d_{22}^2+4d_{14}^2}},$ and  hence, its corresponding azimuth angle   $\varphi^*$ and  phase angle $\theta^*$,$$\varphi^*=\phi\mbox{ or }\phi+\frac{2\pi}3\mbox{ or } \phi+\frac{4\pi}3, \phi=\frac13\arcsin\frac{d_{11}}{\sqrt{d_{11}^2+d_{22}^2}},$$
$$d_{14}>0 \Longrightarrow \theta^*=\pi-\arccos\sqrt{\frac{1+\alpha}2}\mbox{ or }2\pi-\arccos\sqrt{\frac{1+\alpha}2};$$ 
$$d_{14}<0 \Longrightarrow  \theta^*= \arccos\sqrt{\frac{1+\alpha}2}\mbox{ or }\pi+\arccos\sqrt{\frac{1+\alpha}2}.$$

\begin{example}[Point Groups 4mm and 6mm (e.g., BaTiO$_3$]\label{ex:7}   In Point Groups 4mm and 6mm Crystal, there are three non-zero independent components of second order susceptibility tensor ${\bm\chi}^{(2)}=(t_{ijk})$, \ $t_{322}=t_{311}=d_{31}, t_{223}=t_{232}=t_{113}=t_{131}=d_{15}, t_{333}=d_{33} .$
\end{example}

Obviously, by Eqs. \eqref{eq:34} and \eqref{eq:38}, we have $$a_{11}({\bm x})=a_{22}({\bm x})=a_{12}({\bm x})=0, V_1({\bm x})= V_2({\bm x})=0, V_3({\bm x})=d_{31}(x_1^2+x^2_2).$$
Then from Eqs. \eqref{eq:37},  it follows that $$ \theta=0,\pi \Longrightarrow {\bm \chi}_{\text{eff}}^{oo-e}{\bm y}{\bm x}{\bm x}=0.$$ 

It follows from Eqs. \eqref{eq:35} and \eqref{eq:36} that $${\bm \chi}_1^{\textbf{eff}}=\max\limits_{\|{\bm x}\|_2=1}\sqrt{0+(d_{31}(x_1^2+x^2_2))^2}
=\max_{\|{\bm x}\|_2=1}|d_{31}(x_1^2+x^2_2)|=|d_{31}|
$$ with its maximizer $${\bm x}^*=(\sin\varphi^*, -\cos\varphi^*,0), \varphi\in[0,2\pi],  {\bm z}^*=(\cos\theta^*,\sin\theta^*)=\begin{cases}
	(0,1), d_{31}>0,\\
	(0,-1), d_{31}<0,
\end{cases}$$ and  hence, its corresponding azimuth angle $\varphi^*$ and  phase angle $\theta^*$, $$\varphi^*\in[0,2\pi],\mbox{ and } \theta^*=\frac\pi2 \ \ (d_{31}>0),\frac{3\pi}2  \ \  (d_{31}<0).$$  

\begin{example}[Point Group 32 (e.g., $\alpha$-SiO$_2$, HgO)] \label{ex:8}   In Point Group 32 Crystal, there are five non-zero independent components of second order susceptibility tensor ${\bm\chi}^{(2)}=(t_{ijk})$, $$t_{111}=-t_{122}=-t_{221}=-t_{212}=d_{11}, t_{123}=t_{132}=-t_{213}=-t_{231}=d_{14}. $$
\end{example}

Obviously, by Eqs. \eqref{eq:34} and \eqref{eq:38}, we have $$a_{11}({\bm x})=d_{11}(3x_1x_2^2-x_1^3), a_{22}({\bm x})=0, a_{12}({\bm x})=-d_{14}(x_1^2+x^2_2),$$ $$V_1({\bm x})=d_{11}(x_1^2-x_2^2),  V_2({\bm x})=-2d_{11}x_1x_2, V_3({\bm x})=0.$$
Then from Eqs. \eqref{eq:37} and \eqref{eq:311},  it follows that  $$ \theta=\frac\pi2, \frac{3\pi}2\Longrightarrow {\bm \chi}_{\text{eff}}^{oo-e}{\bm y}{\bm x}{\bm x}=0,$$ $$\tan\theta=\frac{d_{11}\sin3\varphi}{2d_{14}}\Longrightarrow {\bm \chi}_{\text{eff}}^{ee-o}{\bm x}{\bm y}{\bm y}=0.$$ It is easy to show that $$ x_2V_1({\bm x})-x_1V_2({\bm x})=d_{11}\cos3\varphi=0, \mbox{i.e., } \varphi= \frac\pi6, \frac{\pi}2, \frac{3\pi}2,  \frac{5\pi}6, \frac{7\pi}6,\frac{11\pi}6 \Longrightarrow {\bm \chi}_{\text{eff}}^{oo-e}{\bm y}{\bm x}{\bm x}=0.$$ 

It follows from Eqs. \eqref{eq:35} and \eqref{eq:36} that  $${\bm \chi}_1^{\textbf{eff}}=\max\limits_{\|{\bm x}\|_2=1}\sqrt{(d_{11}(3x_1^2x_2-x^3_2))^2}
=\max\limits_{\|{\bm x}\|_2=1}|d_{11}(3x_1^2x_2-x^3_2)|=|d_{11}|
$$ with its maximizer ${\bm x}^*=(\sin\varphi^*, -\cos\varphi^*,0)$ with $$ \cos3\varphi^*=\pm1,\  {\bm z}^*=(\cos\theta^*,\sin\theta^*)=\begin{cases}
	(1,0),\ \ \cos3\varphi=1, d_{11}>0,\\
	(-1,0), \cos3\varphi=1, d_{11}<0,\\
	(1,0),\ \  \cos3\varphi=-1, d_{11}<0,\\
	(-1,0), \cos3\varphi=-1, d_{11}>0,
\end{cases}$$ and  hence,  its corresponding azimuth angle $\varphi^*$ and  phase angle $\theta^*$,   $$d_{11}>0 \Longrightarrow \begin{cases}
	\varphi^*=0\mbox{ or }\frac{2\pi}3\mbox{ or }\frac{4\pi}3, \theta^*=0;  \\ \varphi^*=\frac\pi3\mbox{ or }\pi\mbox{ or }\frac{5\pi}3,\theta^*=\pi;
\end{cases}$$  $$d_{11}<0 \Longrightarrow \begin{cases}
	\varphi^*=\pi\mbox{ or }\frac{2\pi}3\mbox{ or }\frac{4\pi}3, \theta^*=\pi; \\ \varphi^*=\frac\pi3\mbox{ or }\pi\mbox{ or }\frac{5\pi}3, \theta^*=0.
\end{cases}$$

It follows from Eqs. \eqref{eq:39} and \eqref{eq:310} that $$ \aligned {\bm \chi}_2^{\textbf{eff}}=&\frac12\max\limits_{\|{\bm x}\|_2=1}(d_{11}(3x_1x_2^2-x_1^3)+\sqrt{(d_{11}(3x_1x_2^2-x_1^3))^2+4(-d_{14}(x_1^2+x^2_2))^2})\\=&\frac12\left(|d_{11}|+\sqrt{d^2_{11}+4d^2_{14}}\right)\endaligned$$ with its maximizer ${\bm x}^*=(\sin\varphi^*, -\cos\varphi^*,0)$ with $$ \sin3\varphi^*=\begin{cases}
	1, d_{11}>0,\\
	-1, d_{11}<0,
\end{cases} {\bm z}^*=(\cos\theta^*,\sin\theta^*)=\begin{cases}
	\pm\left(\sqrt{\frac{1+\alpha}2},\sqrt{\frac{1-\alpha}2}\right), d_{14}<0,\\
	\pm\left(\sqrt{\frac{1+\alpha}2},-\sqrt{\frac{1-\alpha}2}\right), d_{14}>0,
\end{cases}$$ where $\alpha=\frac{|d_{11}|}{\sqrt{d_{11}^2+4d_{14}^2}},$ and  hence, its corresponding azimuth angle $\varphi^*$ and  phase angle $\theta^*$,  $$\varphi^*=\frac{\pi}6 \mbox{ or }\frac{5\pi}6\mbox{ or } \frac{3\pi}2(d_{11}>0);   \varphi^*=\frac{\pi}2 \mbox{ or }\frac{7\pi}6 \mbox{ or } \frac{11\pi}6\ (d_{11}<0),$$
$$d_{14}>0 \Longrightarrow \theta^*=\pi-\arccos\sqrt{\frac{1+\alpha}2}\mbox{ or }2\pi-\arccos\sqrt{\frac{1+\alpha}2};$$
$$d_{14}<0 \Longrightarrow  \theta^*= \arccos\sqrt{\frac{1+\alpha}2}\mbox{ or }\pi+\arccos\sqrt{\frac{1+\alpha}2}.$$

\begin{example}[Point Groups 422 and 622  (e.g.,  SiO$_2$)] \label{ex:9}   In Point Groups 422 and 622 Crystal, there is one non-zero independent components of second order susceptibility tensor ${\bm\chi}^{(2)}=(t_{ijk})$, $t_{123}=t_{132}=-t_{231}=-t_{213}=d_{14}. $
\end{example}

Obviously, by Eqs. \eqref{eq:34} and \eqref{eq:38}, we have $$a_{11}({\bm x})= a_{22}({\bm x})=0, a_{12}({\bm x})=-d_{14}(x_1^2+x^2_2),  V_1({\bm x})=V_2({\bm x})=V_3({\bm x})=0.$$
Then from Eq.  \eqref{eq:311},  it follows that  $${\bm \chi}_{\text{eff}}^{oo-e}{\bm y}{\bm x}{\bm x}=0, \theta=0,\pi,\frac\pi2, \frac{3\pi}2\Longrightarrow {\bm \chi}_{\text{eff}}^{ee-o}{\bm x}{\bm y}{\bm y}=0.$$ 

It follows from Eqs. \eqref{eq:39} and \eqref{eq:310} that $$d_2^{\textbf{eff}}=\frac12\max\limits_{\|{\bm x}\|_2=1}(0+\sqrt{4(-d_{14}(x_1^2+x^2_2))^2})=|d_{14}|$$ with its maximizer ${\bm x}^*=(\sin\varphi^*, -\cos\varphi^*,0),  \varphi\in[0,2\pi],$ $${\bm z}^*=(\cos\theta^*,\sin\theta^*)=\begin{cases}
	\pm\left(\sqrt{\frac12},\sqrt{\frac12}\right), d_{14}<0,\\
	\pm\left(\sqrt{\frac12},-\sqrt{\frac12}\right), d_{14}>0,
\end{cases}$$ and  hence, its corresponding azimuth angle $\varphi$ and  phase angle $\theta^*$,  $$\varphi^*\in[0,2\pi], \theta^*=\frac{\pi}4 \mbox{ or }\frac{5\pi}4\  (d_{11}<0);  \frac{3\pi}4 \mbox{ or } \frac{7\pi}4 \ (d_{11}>0).$$

\begin{example}[Point Groups $\bar{6}$m2  (e.g., BaTiSi$_3$O$_9$)] \label{ex:10}   In Point Group $\bar{6}$m2 Crystal, there is one non-zero independent components of second order susceptibility tensor ${\bm\chi}^{(2)}=(t_{ijk})$, $t_{222}=-t_{211}=-t_{121}=-t_{112}=d_{22}. $
\end{example}

Obviously, by Eqs. \eqref{eq:34} and \eqref{eq:38}, we have $$a_{11}({\bm x})=d_{22}(3x^2_1x_2-x_2^3), a_{22}({\bm x})=a_{12}({\bm x})=0, $$ $$V_1({\bm x})=-2d_{22}x_1x_2, V_2({\bm x})=d_{22}(x_2^2-x_1^2), V_3({\bm x})=0.$$
Then from Eqs. \eqref{eq:37} and \eqref{eq:311},  it follows that $$\theta=\frac\pi2, \frac{3\pi}2\Longrightarrow {\bm \chi}_{\text{eff}}^{oo-e}{\bm y}{\bm x}{\bm x}=0;$$ $$\theta=\frac\pi2, \frac{3\pi}2\Longrightarrow {\bm \chi}_{\text{eff}}^{ee-o}{\bm x}{\bm y}{\bm y}=0.$$  It is easy to show that $$ x_2V_1({\bm x})-x_1V_2({\bm x})=-d_{22}\sin3\varphi=0, \mbox{i.e., } \varphi=0, \frac\pi3, \frac{2\pi}3, \pi, \frac{4\pi}3, \frac{5\pi}3 \Longrightarrow {\bm \chi}_{\text{eff}}^{oo-e}{\bm y}{\bm x}{\bm x}=0, $$ $$ a_{11}({\bm x})=d_{22}\cos3\varphi=0, \mbox{i.e., } \varphi=\frac{\pi}6, \frac{\pi}2, \frac{3\pi}2, \frac{5\pi}6, \frac{7\pi}6, \frac{11\pi}6 \Longrightarrow {\bm \chi}_{\text{eff}}^{ee-o}{\bm x}{\bm y}{\bm y}=0.$$

It follows from Eqs. \eqref{eq:35} and \eqref{eq:36} that $${\bm \chi}_1^{\textbf{eff}}=\max\limits_{\|{\bm x}\|_2=1}\sqrt{d^2_{22}(x_1^3-3x_1x^2_2)^2}=\max\limits_{\|{\bm x}\|_2=1}|d_{22}(x_1^3-3x_1x^2_2)|=|d_{22}|$$
with its maximizer ${\bm x}^*=(\sin\varphi^*, -\cos\varphi^*,0) $ with $$ \sin3\varphi^*=\pm1, {\bm z}^*=(\cos\theta^*,\sin\theta^*)=\pm(1,0),$$  and  hence, its corresponding azimuth angle $\varphi^*$ and  phase angle $\theta^*$,   $$\varphi^*=\frac{\pi}6 \mbox{ or }\frac{5\pi}6 \mbox{ or }\frac{3\pi}2, \theta^*=0\ (d_{22}<0), \pi\ (d_{22}>0);$$ $$\varphi^*= \frac{\pi}2\mbox{ or } \frac{7\pi}6 \mbox{ or }\frac{11\pi}6, \theta^*=0\ (d_{22}>0), \pi\ (d_{22}<0).$$

It follows from Eqs. \eqref{eq:39} and \eqref{eq:310} that  $$\aligned{\bm \chi}_2^{\textbf{eff}}=&\frac12\max\limits_{\|{\bm x}\|_2=1}(a_{11}({\bm x})+\sqrt{(a_{11}({\bm x}))^2})=\max\limits_{\|{\bm x}\|_2=1}a_{11}({\bm x})\\=&\max\limits_{\|{\bm x}\|_2=1}d_{22}(3x^2_1x_2-x_2^3)=|d_{22}|\endaligned$$ with its maximizer ${\bm x}^*=(\sin\varphi^*, -\cos\varphi^*,0)$ with $$ \cos3\varphi^*=
1 \ (d_{22}>0),\ -1\ (d_{22}<0), {\bm z}^*=(\cos\theta^*,\sin\theta^*)=	\pm(1,0),$$ and  hence, its corresponding azimuth angle $\varphi^*$ and  phase angle $\theta^*$ , $$d_{22}>0\Longrightarrow\varphi^*=0\mbox{ or }\frac{2\pi}3  \mbox{ or }\frac{4\pi}3, \theta^*=0\mbox{ or }\pi;$$ $$d_{22}<0\Longrightarrow\varphi^*=\frac\pi3\mbox{ or }\pi  \mbox{ or }\frac{5\pi}3, \theta^*=0\mbox{ or }\pi.$$

\end{document}